\documentclass[letterpaper, 10pt, conference]{ieeeconf} 

\IEEEoverridecommandlockouts 
\usepackage{amsfonts,amsmath,amssymb}
\usepackage{lscape,graphicx}
\newtheorem{theorem}{\indent Theorem}[section]
\newtheorem{lemma}[theorem]{\indent Lemma}
\newtheorem{corollary}[theorem]{\indent Corollary}
\newtheorem{remark}[theorem]{\indent Remark}

\newtheorem{proposition}[theorem]{\indent Proposition}

\newtheorem{EXAMPLE}{\indent Example}[section]

\newcommand{\define}{\stackrel{\mbox{\tiny $\triangle$}}{=}}

\newcommand{\AWGNC}{\mbox{\tiny AWGNC}}
\newcommand{\BSC}{\mbox{\tiny BSC}}
\newcommand{\BEC}{\mbox{\tiny BEC}}
\newcommand{\maxfrac}{\mbox{\tiny max-frac}}

\newcommand{\code}{\mathcal{C}}

\newcommand{\cL}{{\mbox{\boldmath $L$}}}

\newcommand{\cHsmall}{{\mbox{\scriptsize \boldmath $H$}}}
\newcommand{\cH}{{\mbox{\boldmath $H$}}}
\newcommand{\bldH}{{\mbox{\boldmath $H$}}}

\newcommand{\cC}{{\mathcal{C}}}

\newcommand{\cI}{{\mathcal{I}}}
\newcommand{\cJ}{{\mathcal{J}}}

\newcommand{\cK}{{\mathcal{K}}}

\newcommand{\cS}{{\mathcal{S}}}

\newcommand{\bH}{{\mbox{\boldmath $H$}}}

\newcommand{\w}{{\mathsf{w}}}
\newcommand{\weight}{{\mathsf{w}}}

\newcommand{\ff}{{\mathbb{F}}}

\newcommand\rr{{\mathbb R}}

\newcommand\nn{{\mathbb N}}

\newcommand{\bldc}{{\mbox{\boldmath $c$}}}

\newcommand{\bldh}{{\mbox{\boldmath $h$}}}
\newcommand{\bldi}{{\mbox{\boldmath $i$}}}

\newcommand{\bldx}{{\mbox{\boldmath $x$}}}
\newcommand{\bldxx}{{\mbox{\scriptsize \boldmath $x$}}}

\newcommand{\ones}{{\mbox{\boldmath $1$}}}
\newcommand{\smallzeros}{{\mbox{\scriptsize \boldmath $0$}}}
\newcommand{\zeros}{{\mbox{\boldmath $0$}}}

\newcommand{\bldzero}{{\mbox{\boldmath $0$}}}

\newcommand{\GL}{\operatorname{GL}}

    \def\squarebox#1{\hbox to #1{\hfill\vbox to #1{\vfill}}}

\newcommand{\mat}[1]{\left[\ \ \begin{matrix}#1\end{matrix}\;\ \right]}
\newcommand{\zo}{\mbox{\footnotesize \!$0$\!\! \normalsize}}
\newcommand{\ze}{\mbox{\footnotesize \!$\mathbf{1}$\!\! \normalsize}}

\newlength{\Algwidth}

\title{\LARGE\bf Exploration of AWGNC and BSC Pseudocodeword Redundancy}

\author{Jens Zumbr\"agel, Mark F. Flanagan, and Vitaly Skachek%
  \thanks{This work was supported in part by the Claude Shannon
    Institute for Discrete Mathematics, Coding and Cryptography
    (Science Foundation Ireland Grant 06/MI/006). The work of
    V. Skachek was done in part while he was with the Claude Shannon
    Institute, University College Dublin. His work was also supported
    in part by the National Research Foundation of Singapore (Research
    Grant NRF-CRP2-2007-03).}
  \thanks{J. Zumbr\"agel and M.F. Flanagan are with the Claude Shannon
    Institute, University College Dublin, Belfield, Dublin 4, Ireland.
    Emails: {\tt\small jens.zumbragel@ucd.ie}, {\tt\small
      mark.flanagan@ieee.org}.}%
  \thanks{V. Skachek is with Division of Mathematical Sciences, School
    of Physical and Mathematical Sciences, Nanyang Technological
    University, 21 Nanyang Link, 637371 Singapore.  Email: {\tt\small
      vitaly.skachek@ntu.edu.sg}.}%
}

\begin{document}


\maketitle
\thispagestyle{empty}
\pagestyle{empty}

\begin{abstract}
  The AWGNC, BSC, and max-fractional \emph{pseudocodeword redundancy}
  $\rho(\code)$ of a code $\code$ is defined as the smallest number of
  rows in a parity-check matrix such that the corresponding minimum
  pseudoweight is equal to the minimum Hamming distance of
  $\code$. This paper provides new results on the AWGNC, BSC, and
  max-fractional pseudocodeword redundancies of codes. The
  pseudocodeword redundancies for all codes of small length (at most
  $9$) are computed. Also, comprehensive results are provided on the
  cases of cyclic codes of length at most $250$ for which the
  eigenvalue bound of Vontobel and Koetter is sharp.
\end{abstract}


\section{Introduction}

Pseudocodewords play a significant role in the finite-length analysis
of binary linear low-density parity-check (LDPC) codes under
linear-programming (LP) or message-passing (MP) decoding (see
e.g.~\cite{KV-characterization, KV-long-paper}).  The concept of
\emph{pseudoweight} of a pseudocodeword was introduced in~\cite{FKKR}
as an analog to the pertinent parameter in the maximum likelihood (ML)
decoding scenario, i.e. the signal Euclidean distance in the case of
the additive white Gaussian noise channel (AWGNC), or the Hamming
distance in the case of the binary symmetric channel (BSC).
Accordingly, for a binary linear code $\code$ and a parity-check
matrix $\bldH$ of $\code$, the (AWGNC or BSC) minimum pseudoweight
$\w^{\min}(\bldH)$ may be considered as a first-order measure of
decoder error-correcting performance for LP or MP decoding. Note that
$\w^{\min}(\bldH)$ may be different for different matrices~$\bldH$:
adding redundant rows to~$\bldH$ introduces additional constraints on
the so-called \emph{fundamental cone} and may thus increase the
minimum pseudoweight.  Another closely related measure is the
max-fractional weight (pseudoweight).  It serves as a lower bound on
both AWGNC and BSC pseudoweights.

The AWGNC (or BSC) pseudocodeword redundancy $\rho_{\AWGNC}(\code)$
(or $\rho_{\BSC}(\code)$, respectively) of a code $\code$ is defined
as the minimum number of rows in a parity-check matrix $\bldH$ such
that the corresponding minimum pseudoweight $\w^{\min}(\bldH)$ is as
large as its minimum Hamming distance $d$. It is set to infinity if
there is no such matrix. We sometimes simply write $\rho(\code)$, when
the type of the channel is clear from the context.

The pseudocodeword redundancy for the binary erasure channel (BEC),
$\rho_{\BEC}(\code)$, was studied in~\cite{Schwartz-Vardy}, where it
was shown to be finite for any binary linear code $\code$. The authors
also presented some bounds on $\rho_{\BEC}(\code)$ for general linear
codes, and for some specific families of codes.  The study of BSC
pseudoredundancy was initiated in~\cite{Kelley-Sridhara}, where the
authors presented bounds on $\rho_{\BSC}(\code)$ for various families
of codes.  In a recent work~\cite{ISIT-Paper}, we provided some bounds
on $\rho_{\AWGNC}(\code)$ and $\rho_{\BSC}(\code)$ for general linear
codes. In particular, \cite{ISIT-Paper} listed some preliminary
results regarding the AWGNC and BSC pseudocodeword redundancies of
short codes; this paper provides more comprehensive results in this
direction.

The outline of the paper is as follows.  After providing detailed
definitions in Section~\ref{sec:general_settings} we prove several new
theoretical results on the pseudocodeword redundancy in
Sections~\ref{sec:basic_results} and~\ref{sec:rows_weight_2}.  The
next two sections are devoted to experimental results;
Section~\ref{sec:short_codes} examines the pseudocodeword redundancy
for all codes of small length, and Section~\ref{sec:kv_bound} deals
with cyclic codes that meet the eigenvalue bound of Vontobel and
Koetter.


\section{General Settings}\label{sec:general_settings}

Let $\code$ be a code of length $n \in \nn$ over the binary field
$\ff_2$, defined by
\begin{equation}\label{eq:code_definition}
  \code = \ker \cH
  = \{ \bldc \in \ff_2^n \; : \; \cH \bldc^T = \bldzero^T \}
\end{equation}
where $\cH$ is an $m \times n$ \emph{parity-check matrix} of the code
$\code$. Obviously, the code $\code$ may admit more than one
parity-check matrix, and all the codewords form a linear vector space
of dimension $k \ge n-m$. We say that $k$ is the \emph{dimension} of
the code $\code$.  We denote by $d(\code)$ (or just $d$) the minimum
Hamming distance (also called the minimum distance) of $\code$. The
code $\code$ may then be referred to as an $[n,k,d]$ linear code over
$\ff_2$. 


The parity-check matrix $\cH$ is said to be $(w_c,w_r)$-regular if
every column of $\cH$ has exactly $w_c$ nonzero symbols, and every row
of it has exactly $w_r$ nonzeros. The matrix $\cH$ is called
$w$-regular if every row and every column in it has $w$ nonzeros.

Denote the set of column indices and the set of row indices of $\cH$
by $\cI = \{1, 2, \dots, n \}$ and $\cJ = \{1, 2, \dots, m \}$,
respectively. For $j \in \cJ$, we denote $\cI_j \define \{ i \in \cI
\; : \; H_{j,i} \neq 0 \}$, and for $i \in \cI$, we denote $\cJ_i
\define \{ j \in \cJ \; : \; H_{j,i} \neq 0 \}$. The \emph{fundamental
  cone} of $\cH$, denoted $\cK(\cH)$, is defined
in~\cite{Feldman-et-al} and~\cite{KV-long-paper} as the set of vectors
$\bldx \in \rr^n$ that satisfy
\begin{equation}\label{eq:polytope-inequality-1}
  \forall j \in \cJ, \; \forall \ell \in \cI_j  \; : 
  \; x_\ell \le \sum_{i \in \cI_j \backslash \{ \ell \}} x_i \; ,
\end{equation}
\begin{equation}\label{eq:polytope-inequality-2}
  \forall i \in \cI \; : \; x_i \ge 0 \; .
\end{equation}

The vectors $\bldx\in\rr^n$
satisfying~(\ref{eq:polytope-inequality-1})
and~(\ref{eq:polytope-inequality-2}) are called \emph{pseudocodewords}
of $\code$ with respect to the parity-check matrix $\cH$.  Note that
the fundamental cone $\cK(\cH)$ depends on the parity-check matrix
$\cH$ rather than on the code $\code$ itself.  At the same time, the
fundamental cone is independent of the underlying communication
channel.

The BEC, AWGNC, BSC pseudoweights and max-fractional weight of a
nonzero pseudocodeword $\bldx \in \cK(\cH)$ were defined
in~\cite{FKKR} and~\cite{KV-long-paper} as follows:
\begin{align*}
  \weight_{\BEC} (\bldx) & \,\define\, 
  \left| \mbox{supp} ( \bldx ) \right| \; , \\
  \weight_{\AWGNC} (\bldx) & \,\define\,  
  \frac{\left( \sum_{i \in \cI} x_i \right)^2}{\sum_{i \in \cI} x_i^2}\;.
\end{align*}
Let $\bldx'$ be a vector in $\rr^n$ with the same components as
$\bldx$ but in non-increasing order.  For $i-1 < \xi \le i$, where $1
\le i \le n$, let $\phi(\xi) \stackrel{\triangle}{=} x'_i$. Define
$\Phi(\xi) \define \int_{0}^{\xi} \phi(\xi') \; d \xi'$ and \[
\weight_{\BSC}(\bldx) \define 2\, \Phi^{-1} ( \Phi(n)/2 ) \; . \] 
Finally, the max-fractional weight of $\bldx$ is defined as
\[\weight_{\maxfrac} (\bldx) \,\define\,
\frac{\sum_{i \in \cI} x_i}{\max_{i \in \cI} x_i} \; .\]

We define the BEC \emph{minimum pseudoweight} of the code $\code$ with
respect to the parity-check matrix $\cH$ as
\[ \weight_{\BEC}^{\min} (\cH) \,\define\, \min_{\bldxx \in \cK(\cHsmall)
  \backslash \{ \smallzeros \} } \weight_{\BEC} (\bldx) \; .\] The
quantities $\weight_{\AWGNC}^{\min} (\cH) $, $\weight_{\BSC}^{\min}
(\cH) $ and $\weight_{\maxfrac}^{\min} (\cH)$ are defined
similarly. 
When the type of pseudoweight is clear from the context, 
we might use the notation $\weight^{\min} (\cH)$.
Note that all four minimum pseudoweights are upper bounded
by~$d$, the code's minimum distance.

Then we define the BEC \emph{pseudocodeword redundancy} of the code
$\code$ as
\[ \rho_{\BEC}(\code) \,\define\, \inf\{\#\text{rows}(\cH) \mid
\ker\cH=\code\,,\, \weight_{\BEC}^{\min}(\cH)=d\} \: ,\] where
$\inf\varnothing\define\infty$, and similarly we define the
pseudocodeword redundancies $\rho_{\AWGNC}(\code)$,
$\rho_{\BSC}(\code)$ and $\rho_{\maxfrac}(\code)$ for the AWGNC and
BSC pseudoweights, and the max-fractional weight.  When the type of
pseudocodeword redundancy is clear from the context, we might use the
notation $\rho(\code)$. We remark that all pseudocodeword redundancies
satisfy $\rho(\code) \ge r\define n-k$.

We describe the behavior of the pseudocodeword redundancy and the
minimum pseudoweight for a given binary linear $[n,k,d]$ code~$\cC$
by introducing four classes of codes:\smallskip

\begin{description}\setlength{\parsep}{0pt}\setlength{\itemsep}{4pt}
\item[{\bf (class 0)}] \quad $\rho(\cC)$ is infinite, i.e.\ there is
  no parity-check matrix $\bH$ with $d=\w^{\min}(\bH)$,
\item[{\bf (class 1)}] \quad $\rho(\cC)$ is finite, but $\rho(\cC)>r$,
\item[{\bf (class 2)}] \quad $\rho(\cC)=r$, but $\cC$ is not in
  class 3,
\item[{\bf (class 3)}] \quad $d=\w^{\min}(\bH)$ for \emph{every}
  parity-check matrix $\bH$ of $\cC$.
\end{description}


\section{Basic Results}\label{sec:basic_results}

The next lemma is taken from~\cite{KV-long-paper}.

\begin{lemma}\label{lemma:relations}
  Let $\code$ be a binary linear code with the parity-check matrix
  $\cH$.  Then,
  \begin{eqnarray*}
    &\weight_{\maxfrac}^{\min} (\cH) \; \le \; 
    \weight_{\AWGNC}^{\min} (\cH) \; \le \; \weight_{\BEC}^{\min} (\cH) \; , \\
    &\weight_{\maxfrac}^{\min} (\cH) \; \le \; 
    \weight_{\BSC}^{\min} (\cH) \; \le \; \weight_{\BEC}^{\min} (\cH) \; . 
  \end{eqnarray*}
\end{lemma}\smallskip

The following theorem is a straightforward corollary to
Lemma~\ref{lemma:relations}.

\begin{theorem}\label{thrm:pseudoredundancies}
  Let $\code$ be a binary linear code.  Then,
  \begin{eqnarray*}
    &\rho_{\maxfrac} (\code) \; \ge \; \rho_{\AWGNC} (\code) \; \ge \; 
    \rho_{\BEC} (\code) \; , \\
    &\rho_{\maxfrac} (\code) \; \ge \; \rho_{\BSC} (\code) \; \ge \; 
    \rho_{\BEC} (\code) \; .
  \end{eqnarray*}
\end{theorem}\smallskip

The following results hold with respect to the AWGNC and BSC
pseudoweights, and the max-fractional weight.

\begin{lemma}\label{lemma:puncturing}
  Let $\cC$ be an $[n,k,d]$ code having $t$ zero coordinates, and let
  $\cC'$ be the $[n-t,k,d]$ code obtained by puncturing $\cC$ at these
  coordinates.  Then 
  \[ \rho(\cC')\le \rho(\cC)\le \rho(\cC')+t \;. \]
\end{lemma}\smallskip

In the proof we use the following notation: We identify $\rr^n$ with
$\rr^{\cI}$, and for $x\in\rr^{\cI}$ and some subset
$\cI'\subseteq\cI$ we let $\bldx|_{\cI'}\in\rr^{\cI'}$ be the
projection of $\bldx$ onto the coordinates in $\cI'$.

\begin{proof}
  Let $\cI'\subseteq\cI$ be the set of nonzero coordinates of the code
  $\cC$.  To prove the first inequality, let $\bH$ be a $\rho\times n$
  parity-check matrix for $\cC$.  Consider its $\rho\times(n-t)$
  submatrix $\bH'$ consisting of the columns corresponding to $\cI'$.
  Then $\bH'$ is a parity-check matrix for $\cC'$, and \[\cK(\bH') =
  \{\bldx|_{\cI'} \,:\, \bldx\in\cK(\bH),\
  \bldx|_{\cI\setminus\cI'}=\bldzero\} \;.\] Therefore,
  $\w^{\min}(\bH')\ge \w^{\min}(\bH)$, and this proves
  $\rho(\cC')\le\rho(\cC)$.

  For the second inequality, let $\bH'$ be a $\rho'\times(n-t)$
  parity-check matrix for $\cC'$.  Now we consider a $(\rho'+t)\times
  n$ matrix $\bH$ with the following properties: The upper
  $\rho'\times n$ submatrix of $\bH$ consists of the columns of $\bH'$
  at positions $\cI'$ and of zero-columns at positions
  $\cI\setminus\cI'$, and the lower $t\times n$ submatrix consists of
  rows of weight~$1$ that have $1$s at the positions
  $\cI\setminus\cI'$.  Then $\cC=\ker\bH$ and
  \[ \cK(\bH) = \{ \bldx\in\rr^\cI\,:\, \bldx|_{\cI'}\in\cK(\bH'),\
  \bldx|_{\cI\setminus\cI'}=\bldzero\}\;. \] Consequently, $\w^{\min}(\bH) =
  \w^{\min}(\bH')$, and this proves $\rho(\cC)\le \rho(\cC')+t$.
\end{proof}

\begin{lemma}\label{lemma:distance_two}
  Let $\cC$ be a code of minimum distance $d\le 2$.  Then
  $d=\w^{\min}(\bH)$ for any parity-check matrix $\bH$ of $\cC$,
  i.e.\ $\cC$ is in class $3$ (for AWGNC and
  BSC pseudoweight, and for max-fractional weight).
\end{lemma}

\begin{proof}
  By Lemma~\ref{lemma:relations} it suffices to prove this lemma for
  the max-fractional weight $\w=\w_{\maxfrac}$.  Since $\w(\bldx)\ge
  1$ holds for all nonzero pseudocodewords, we always have
  $\w^{\min}(\bH)\ge 1$, which proves the result in the case $d=1$.

  Let $d=2$ and $\bH$ be a parity-check matrix for $\cC$.  Let
  $\bldx\in\cK(\bH)$ and let $x_{\ell}$ be the largest coordinate.
  Since $d=2$ there is no zero column in $\bH$ and thus there exists a
  row $j$ with $\ell\in\cI_j$.  Then $x_{\ell}\le
  \sum_{i\in\cI\setminus\{\ell\}}x_i$, hence
  $2x_{\ell}\le\sum_{i\in\cI}x_i$, and thus $\w(\bldx)\ge 2$.  It
  follows $\w^{\min}(\bH)\ge 2$ and the lemma is proved.
\end{proof}


\section{Parity-check matrices with rows of
  weight~$2$}\label{sec:rows_weight_2}

The main result of this section appears in the following lemma. 

\begin{lemma}\label{lemma:row_weight_2}
  Let $\cH$ be a parity-check matrix of $\code$ such that every row in
  $\cH$ has weight~$2$.  Then:
  \begin{enumerate}
  \item[(a)] There is an equivalence relation on the set $\cI$ of column
    indices of $\cH$ such that for a vector $\bldx\in\rr^n$ with
    non-negative coordinates we have $\bldx\in\cK(\cH)$ if and only if
    $\bldx$ has equal coordinates within each equivalence class.
  \item[(b)] The minimum distance of $\code$ is equal to its minimum AWGNC
    and BSC pseudoweights and its max-fractional weight with respect
    to $\cH$, i.e.\ $d(\code) = \weight^{\min}(\cH)$.
  \end{enumerate}
\end{lemma}

\begin{proof} 
  For (a), define the required relation $R$ as follows: For
  $i,i'\in\cI$ let $(i,i')\in R$ if and only if $i=i'$ or there exists
  an integer $\ell \ge 1$, column indices $i=i_0, i_1, \dots,
  i_{\ell-1}, i_{\ell}=i'\in\cI$ and row indices $j_1,\dots,j_l\in\cJ$
  such that
    \[ \{ i_0, i_1 \} = \cI_{j_1} \, , \,
    \{ i_1, i_2 \} = \cI_{j_2} \, , \, \dots \, , \,
    \{ i_{\ell-1}, i_{\ell} \} = \cI_{j_{\ell}} \; . \]
  This is an equivalence relation, and it defines equivalence
  classes over $\cI$.  It is easy to check that
  inequalities~(\ref{eq:polytope-inequality-1}) imply that
  $\bldx\in\cK(\cH)$ if and only if $x_i=x_{i'}$ for any $(i,i')\in
  R$.

  In order to prove (b), we note that the minimum (AWGNC, BSC or
  max-fractional) pseudoweight is always bounded above by the minimum
  distance of $\code$, so we only have to show that the minimum
  pseudoweight is bounded below by the minimum distance.
    
  Let $\cS=\{S_1, S_2, \dots, S_t\}$ be the set of equivalence classes
  of $R$, and let $d_S = |S|$ for $S\in\cS$.  It is easy to see that
  the minimum distance of $\code$ is $d = \min_{S\in\cS}d_S$ (since
  the minimum weight nonzero codeword of $\code$ has non-zeros in the
  coordinates corresponding to a set $S\in\cS$ of minimal size and
  zeros everywhere else).

  Now let $\bldx\in\cK(\cH)$.  Since the coordinates $x_i$, $i\in\cI$,
  depend only on the equivalence classes, we may use the notation
  $x_S$, $S\in\cS$.  Let $x_T$ be the largest coordinate. Then:
  \[ \weight_{\maxfrac}(\bldx) = \frac{\sum_{i \in \cI} x_i}{x_T} 
  \ge \frac{\sum_{i\in T}x_i}{x_T} = |T|=d_T \ge d \;. \]
  Therefore, $\weight_{\maxfrac}^{\min} (\cH) \ge d$, and by using
  Lemma~\ref{lemma:relations}, we obtain that $\weight_{\AWGNC}^{\min}
  (\cH) \ge d$ and $\weight_{\BSC}^{\min} (\cH) \ge d$.
\end{proof}\medskip

The following proposition is a stronger version of
Lemma~\ref{lemma:row_weight_2}.

\begin{proposition}\label{prop:row_weight_2}
  Let $\cH$ be an $m \times n$ parity-check matrix of $\code$, and
  assume that $m-1$ first rows in $\cH$ have weight 2.  Denote by
  $\widehat{\cH}$ the $(m-1) \times n$ matrix consisting of these
  rows, consider the equivalence relation of
  Lemma~\ref{lemma:row_weight_2} (a) with respect to $\widehat{\cH}$, and
  assume that $\cI_m$ intersects each equivalence class in at most one
  element.  Then, the minimum distance of $\code$ is equal to its
  minimum AWGNC and BSC pseudoweights and its max-fractional weight
  with respect to $\cH$, i.e.\ $d(\code) = \weight^{\min}(\cH)$.
\end{proposition}

\begin{proof} 
  Let $\cS$ be the set of classes of the aforementioned equivalence
  relation on $\cI$, and let $d_S=|S|$ for $S\in\cS$.  Let
  \[ \cS' = \{ S \in \cS \; : \; |S \cap \cI_m| = 1 \} \;.  \] Also
  let $\cS'' = \cS \backslash \cS'$, so that $S \cap \cI_m =
  \varnothing$ for all $S \in \cS''$.

  Let $\bldx\in\cK(\cH) \backslash \{ \zeros \}$. As before, since the
  coordinates $x_i$, $i\in\cI$, depend only on the equivalence
  classes, we may use the notation $x_S$, $S\in\cS$.  The fundamental
  polytope constraints \eqref{eq:polytope-inequality-1} and
  \eqref{eq:polytope-inequality-2} may then be written as $x_S \ge 0$
  for all $S \in \cS$ and
  \begin{equation}
    \forall R \in \cS' \; : \; x_R \le \sum_{S \in \cS'\setminus\{R\}} x_S \; , 
    \label{eq:polytope_ineq_S'}
  \end{equation}
  respectively, and the max-fractional pseudoweight of $\bldx \in \cK(\cH) \backslash \{ \zeros \}$ is given by 
  \begin{equation}
    \weight_{\maxfrac}(\bldx) 
    = \frac{\sum_{S \in \cS} d_S x_S}{\max_{S \in \cS} x_S} \; .
    \label{eq:new_maxfrac}
\end{equation}

Suppose $\bldx\in\cK(\cH) \backslash \{ \zeros \}$ has minimal
max-fractional pseudoweight.  Let $x_T$ be its largest
coordinate. First note that if there exists $R \in
\cS''\setminus\{T\}$ with $x_R > 0$, setting $x_R$ to zero results in
a new pseudocodeword with lower max-fractional pseudoweight, which
contradicts the assumption that $\bldx$ achieves the
minimum. Therefore $x_R = 0$ for all $R \in \cS''\setminus\{T\}$. We
next consider two cases.

\emph{Case 1:} $T \in \cS''$. If there exists $R \in \cS'$ with $x_R >
0$, setting all such $x_R$ to zero results in a new pseudocodeword
with lower max-fractional pseudoweight, which contradicts the
minimality of the max-fractional pseudoweight of $\bldx$. Therefore
$x_T$ is the only positive coordinate of $\bldx$, and by
\eqref{eq:new_maxfrac} the max-fractional pseudoweight of $\bldx$ is
$d_T$.

\emph{Case 2:} $T \in \cS'$. In this case $x_R = 0$ for all $R \in
\cS''$. From inequality~(\ref{eq:polytope_ineq_S'}) for $R=T$ we
obtain
\[ x_T \le \sum_{S\in\cS'\setminus\{T\}} x_S \;. \]
With $d_0 \define \min_{S\in\cS'\setminus\{T\}}d_S$ it follows that
\[ d_0 x_T \le \sum_{S\in\cS'\setminus\{T\}} d_0 x_S \le
\sum_{S\in\cS'\setminus\{T\}} d_Sx_S \;. \] Consequently, \[
(d_T+d_0)x_T \le \sum_{S\in\cS} d_Sx_S \;, \] and thus
$\weight_{\maxfrac}(\bldx) \ge d_T+d_0$.  We conclude that the minimum
max-fractional pseudoweight is given by
\[ \weight_{\maxfrac}^{\min} (\cH) = \min \left\{ \min_{S, T \in \cS',
    S \neq T} \{d_S + d_T\} \; , \; \min_{S \in \cS''} \{d_S \}
\right\} \;. \] But this is easily seen to be equal to the minimum
distance $d$ of the code.

Finally, by using Lemma~\ref{lemma:relations}, we obtain that
$\weight_{\AWGNC}^{\min} (\cH) = d$ and $\weight_{\BSC}^{\min} (\cH) =
d$.
\end{proof}

\emph{Remark:} Note that the requirement that all $i \in \cI_m$ belong
to the different equivalence classes of $\widehat{\cH}$ in
Proposition~\ref{prop:row_weight_2} is necessary. Indeed, consider the
matrix
\[ \cH = \mat{
  \ze & \ze & \zo & \zo \\
  \zo & \ze & \ze & \zo \\
  \ze & \zo & \ze & \zo \\
  \ze & \ze & \ze & \ze \\
} \;. \]
One can see that there are two equivalence classes for
$\widehat{\cH}$: $S_1 = \{ 1, 2, 3 \}$, $S_2 = \{ 4 \}$. The minimum
distance of the corresponding code $\code$ is $4$ (since $(1, 1, 1, 1)$
is the only nonzero codeword).  However, $\bldx = (1, 1, 1, 3) \in
\cK(\cH)$ is a pseudocodeword of max-fractional weight $2$.

\begin{corollary}\label{corollary:dimension_two}
  Let $\code$ be a code of length~$n$ and dimension~$2$.  Then
  $\rho(\code)=n-2$, i.e.\ $\code$ is of class at least~$2$ (for AWGNC
  and BSC pseudoweight, and for max-fractional weight).
\end{corollary}

\begin{proof}
  We consider two cases. 
  \begin{itemize}
  \item{\it Case 1: $\code$ has no zero coordinates.}
  
    Let $\bldc_1$ and $\bldc_2$ be two linearly independent codewords
    of $\code$.  Define the following subsets of $\cI$:
    \begin{eqnarray*}
      S_1 & \define &  \{ i \in \cI \; : \; 
      i \in \mbox{supp}(\bldc_1) 
      \mbox{ and } i \notin \mbox{supp}(\bldc_2) \} \; \\
      S_2 & \define & 
      \{ i \in \cI \; : \; 
      i \notin \mbox{supp}(\bldc_1) 
      \mbox{ and } i \in \mbox{supp}(\bldc_2) \} \; \\
      S_3 & \define & \{ i \in \cI \; : \; 
      i \in \mbox{supp}(\bldc_1) 
      \mbox{ and } i \in \mbox{supp}(\bldc_2) \} .
    \end{eqnarray*}
    The sets $S_1$, $S_2$ and $S_3$ are pairwise disjoint.  Since
    $\code$ has no zero coordinates, $\cI = S_1 \cup S_2 \cup S_3$.
    The ordering of elements in $\cI$ implies an ordering on the
    elements in each of $S_1$, $S_2$ and $S_3$.  Assume that $S_1 = \{
    i_1, i_2, \cdots, i_{|S_1|}\}$ and $i_1 < i_2 < \cdots <
    i_{|S_1|}$.  If $S_1 \neq \varnothing$, let $m_1 = i_1$ be the
    minimal element in $S_1$, and define an $(|S_1| - 1) \times n$
    matrix $H_1$ as follows:
    \[ (H_1)_{j,\ell} = \left\{ \begin{array}{cl}
        1 & \mbox{ if } i_j = \ell \mbox{ or } i_{j+1} = \ell \; , \\
        & \qquad \qquad j = 1, 2, \cdots, |S_1|-1 \; , \\
        0 & \mbox{ otherwise } \; .  
      \end{array} \right.  \]
    Similarly, define $(|S_2| - 1) \times n$ and $(|S_3| - 1) \times n$
    matrices $H_2$ and $H_3$, with respect to $S_2$ and $S_3$. Let $m_2$
    and $m_3$ be minimal elements of $S_2$ and $S_3$, respectively.
  
    Define also a $1 \times n$ matrix $H_4$: 
    \[ (H_4)_{1,\ell} = \left\{ \begin{array}{cl}
        1 & \mbox{ if } S_j \neq \varnothing \mbox{ and } m_j = \ell \\
        & \qquad \qquad \mbox{ for } j = 1, 2, 3 \; , \\
        0 & \mbox{ otherwise } \; .  
      \end{array} \right.  \]
  
    Finally, define an $(n-2) \times n$ matrix $\cH$ by $\cH^T \define
    [ H_1^T \; | \; H_2^T \; | \;H_3^T \; | \; H_4^T ]$.  (Some of the
    $S_i$'s might be equal to $\varnothing$, in which case the
    corresponding $H_i$ is an $0 \times n$ ``empty'' matrix.)  It is
    easy to see that all rows of $\cH$ are linearly independent, and
    so it is of rank $n-2$.  It is also straightforward that for all
    $\bldc \in \code$ we have $\bldc \in \ker(\cH)$. Therefore, $\cH$
    is a parity-check matrix of~$\code$.
  
    The matrix $\cH$ has a form as in
    Proposition~\ref{prop:row_weight_2} (where $S_1$, $S_2$ and $S_3$
    are corresponding equivalence classes over $\cI$), and therefore
    $\rho(\code) = n-2$.\medskip

  \item{\it Case 2: $\code$ has $t>0$ zero coordinates.}
    
    Consider a code $\code'$ of length $n-t$ obtained by puncturing
    $\code$ in these $t$ zero coordinates. From Case 1 (with respect
    to $\code'$), $\rho(C') = n-t-2$. By applying the rightmost
    inequality in Lemma~\ref{lemma:puncturing}, we have $\rho(C) \le
    n-2$. Since $k=2$, we conclude that $\rho(C) = n-2$.
  \end{itemize}
\end{proof}


\section{The Pseudocodeword Redundancy for Codes of Small
  Length}\label{sec:short_codes}

In this section we compute the AWGNC, BSC, and max-fractional
pseudocodeword redundancies for all codes of small length.  By
Lemma~\ref{lemma:distance_two} it is sufficient to examine only codes
with minimum distance at least~$3$.  Furthermore, in light of
Lemma~\ref{lemma:puncturing} we will consider only codes without
zero coordinates, i.e.\ that have a dual distance of at least~$2$.
Finally, we point out to Corollary~\ref{corollary:dimension_two} for
codes of dimension~$2$, by which we may focus on codes with dimension
at least~$3$.

\subsection{The Algorithm}

To compute the pseudocodeword redundancy of a code $\cC$ we have to
examine all possible parity-check matrices for the code $\cC$, up to
equivalence.  Here, we say that two parity-check matrices $\bH$ and $\bH'$ for
the code $\cC$ are \emph{equivalent} if $\bH$ can be transformed into
$\bH'$ by a sequence of row and column permutations.  In this case,
$\w^{\min}(\bH) = \w^{\min}(\bH')$ holds for the AWGNC and BSC
pseudoweights as well as for the max-fractional weight.
The enumeration of codes and parity-check matrices can be described by
the following algorithm.

\hrulefill\medskip

\textbf{Input:} Parameters $n$ (code length), $k$ (code dimension),
$\rho$ (number of rows of the output parity-check matrices), where
$\rho\ge r\define n-k$.

\textbf{Output:} For all codes of length $n$, dimension $k$,
distance $d\ge 3$, and without zero coordinates, up to code
equivalence: a list of all $\rho\times n$ parity-check matrices, up
to parity-check matrix equivalence.\medskip

\begin{enumerate}
\item Collect the set $X$ of all $r\times n$ matrices such that
  \begin{itemize}
  \item they have different nonzero columns, ordered
    lexicographically,
  \item there is no non-empty $\ff_2$-sum of rows which has weight~$0$
    or~$1$ {\sl (this way, the matrices are of full rank and the
      minimum distance of the row space is at least~$2$)}.
  \end{itemize}
\item Determine the orbits in $X$ under the action of the group
  $\GL_r(2)$ of invertible $r\times r$ matrices over $\ff_2$ {\sl
    (this enumerates all codes with the required properties, up to
    equivalence; the codes are represented by parity-check matrices)}.
\item For each orbit $X_{\cC}$, representing a code $\cC$:
  \begin{enumerate}
  \item Determine the suborbits in $X_{\cC}$ under the action of the
    symmetric group $S_r$ {\sl (this enumerates all parity-check
      matrices without redundant rows, up to equivalence)}.
    
  \item For each representative $\bH$ of the suborbits, collect all
    matrices enlarged by adding $\rho-r$ different redundant rows that
    are $\ff_2$-sums of at least two rows of $\bH$.  Let
    $X_{\cC,\rho}$ be the union of all such $\rho\times n$ matrices.

  \item Determine the orbits in $X_{\cC,\rho}$ under the action of the
    symmetric group $S_{\rho}$, and output a representative for each
    orbit.
  \end{enumerate}
\end{enumerate}

\hrulefill\medskip

This algorithm was implemented in the C programming language.  The minimum pseudoweights for
the various parity-check matrices were computed by using Maple 12 and
the Convex package~\cite{maple-convex}.

\subsection{Results}

We considered all binary linear codes up to length $n$ with distance
$d\ge 3$ and without zero coordinates, up to code equivalence.  The
number of those codes for given length $n$ and dimension $k$ is shown
in Table~\ref{table:no-codes}.

\begin{table}
  \caption{The Number of Binary $[n,k,d]$ Codes
    with~$d\ge 3$~and~without~Zero~Coordinates}
  \label{table:no-codes}
  {\centering
    \begin{tabular}{r|rcccc}
      & $k=1$ & $2$ & $3$ & $4$ & $5$ \\\hline
      & & & & & \vspace{-2mm} \\
      $n=5$ & $1$ & $1$ & & & \\
      $6$ & $1$ & $3$ & $1$ & & \\
      $7$ & $1$ & $4$ & $4$ & $1$ & \\
      $8$ & $1$ & $6$ & $10$ & $5$ & \\
      $9$ & $1$ & $8$ & $23$ & $23$ & $5$
    \end{tabular}\\}%
\end{table}

\subsubsection{AWGNC pseudoweight}

The following results were found to hold for all codes of length $n\le
9$.

\begin{itemize}\setlength{\parsep}{0pt}\setlength{\itemsep}{4pt}
\item There are only two codes $\cC$ with $\rho_{\AWGNC}(\cC)>r$,
  i.e. in class $0$ or $1$ for the AWGNC.
  \begin{itemize}
  \item The $[8,4,4]$ extended Hamming code is the shortest code
    $\cC$ in class 1.  We have $\rho_{\AWGNC}(\cC)=5>4=r$ and out
    of $12$ possible parity-check matrices (up to equivalence) with one
    redundant row there is exactly one matrix $\bH$ with
    $\w^{\min}_{\AWGNC}(\bH)=4$, namely
    \[ \bH = \mat{
      \ze & \zo & \zo & \ze & \ze & \zo & \zo & \ze \\
      \zo & \ze & \zo & \ze & \zo & \ze & \zo & \ze \\
      \zo & \zo & \ze & \ze & \zo & \zo & \ze & \ze \\
      \ze & \ze & \ze & \ze & \zo & \zo & \zo & \zo \\
      \zo & \zo & \zo & \zo & \ze & \ze & \ze & \ze \\
    }\;. \] There is exactly one matrix $\bH$ with
    $\w^{\min}_{\AWGNC}(\bH)=25/7$, and for the remaining matrices $\bH$ we
    have $\w^{\min}_{\AWGNC}(\bH)=3$.
  \item Out of the four $[9,4,4]$ codes there is one code $\cC$ in
    class 1.  We have $\rho_{\AWGNC}(\cC)=6>5=r$ and out of $2526$
    possible parity-check matrices (up to equivalence) with one
    redundant row there are $13$ matrices $\bH$ with
    $\w^{\min}_{\AWGNC}(\bH)=4$.
  \end{itemize}
\item For all codes $\cC$ of minimum distance $d\ge 3$ and for all
  parity-check matrices $\bH$ of $\cC$ we have
  $\w^{\min}_{\AWGNC}(\bH)\ge 3$; in particular, if $d=3$, then $\cC$
  is in class~$3$ for the AWGNC.
\item For the $[7,3,4]$ simplex code there is (up to equivalence) only
  one parity-check matrix $\bH$ without redundant rows such that
  $\w_{\AWGNC}^{\min}(\bH)=4$, namely
  \[ \bH = \mat{
    \ze & \ze & \zo & \ze & \zo & \zo & \zo \\
    \zo & \ze & \ze & \zo & \ze & \zo & \zo \\
    \zo & \zo & \ze & \ze & \zo & \ze & \zo \\
    \zo & \zo & \zo & \ze & \ze & \zo & \ze \\
  }\;. \] It is the only parity-check matrix with constant row weight
  $3$.
\end{itemize}

\subsubsection{BSC pseudoweight}

We computed the pseudocodeword redundancy for the BSC for all codes of
length $n\le 8$.

\begin{itemize}\setlength{\parsep}{0pt}\setlength{\itemsep}{4pt}
\item The shortest codes with $\rho_{\BSC}(\cC)>r$, i.e. in class~$0$
  or $1$ for the BSC, are the $[7,4,3]$ Hamming code $\cC$ and its
  dual code $\cC^{\bot}$, the $[7,3,4]$ simplex code.  We have
  $\rho_{\BSC}(\cC)=4>3$ and $\rho_{\BSC}(\cC^{\bot})=5>4$.
\item There are two codes of length~$8$ with
  $\rho_{\BSC}(\cC)>r$.  These are the $[8,4,4]$ extended Hamming
  code, for which $\rho_{\BSC}(\cC)=6>4$ holds, and one of the three
  $[8,3,4]$ codes, which satisfies $\rho_{\BSC}(\cC)=6>5$.
\end{itemize}

\subsubsection{Max-fractional weight}

We computed the pseudocodeword redundancy with respect to the
max-fractional weight for all codes of length $n\le 8$.

\begin{itemize}\setlength{\parsep}{0pt}\setlength{\itemsep}{4pt}
\item The shortest code with $\rho_{\maxfrac}(\cC)>r$ is the unique
  $[6,3,3]$ code $\cC$.  We have $\rho_{\maxfrac}(\cC)=4>3$.
\item There are two codes of length~$7$ with $\rho_{\maxfrac}(\cC)>r$.
  These are the $[7,4,3]$ Hamming code and the $[7,3,4]$ simplex code,
  which have both pseudocodeword redundancy~$7$.  In both cases, there
  is, up to equivalence, a unique parity-check matrix $\bH$ with seven
  rows that satisfies $d(\cC)=w^{\min}_{\maxfrac}(\bH)$.  

  (This demonstrates that Proposition~5.4 and~5.5 in \cite{ISIT-Paper}
  are sharp for the max-fractional weight, and that the parity-check
  matrices constructed in the proofs are unique in this case.)
\item For the $[8,4,4]$ extended Hamming code $\cC$ we have
  $\rho_{\maxfrac}(C)=\infty$, and thus the code is in class~$0$ for
  the max-fractional weight.  It is the shortest code with infinite
  $\rho_{\maxfrac}(\cC)$.  

  (It can be checked that $\bldx=[1,1,1,1,1,1,1,3]$ is a
  pseudocodeword in $\cK(\bH)$, where the rows of $\bH$ consist of all
  dual codewords; since $\w_{\maxfrac}(\bldx)=\frac{10}{3}<4$, we have
  $w^{\min}_{\maxfrac}(\bH)<4$.)
\item There are two other codes of length~$8$ with
  $\rho_{\maxfrac}(\cC)>r$, namely two of the three $[8,3,4]$ codes,
  having pseudocodeword redundancy~$6$ and~$8$, respectively.
\end{itemize}

\subsubsection{Comparison}

Comparing the results for the AWGNC and BSC pseudoweights, and the
max-fractional weight, we can summarize the results as follows.

\begin{itemize}\setlength{\parsep}{0pt}\setlength{\itemsep}{4pt}
\item For the $[7,4,3]$ Hamming code $\cC$ we have
  $\rho_{\AWGNC}(\cC)=r=3$, $\rho_{\BSC}(\cC)=4$, and
  $\rho_{\maxfrac}(\cC)=7$.
\item For the $[7,3,4]$ simplex code $\cC$ we have
  $\rho_{\AWGNC}(\cC)=r=4$, $\rho_{\BSC}(\cC)=5$, and
  $\rho_{\maxfrac}(\cC)=7$.  
\item For the $[8,4,4]$ extended Hamming code $\cC$ we have
  $\rho_{\AWGNC}(\cC)=5$, $\rho_{\BSC}(\cC)=6$, and
  $\rho_{\maxfrac}(\cC)=\infty$.  This code $\cC$ is the shortest
  one such that $\rho_{\AWGNC}(\cC)>r$, and also the shortest one
  such that $\rho_{\maxfrac}(\cC)=\infty$.
\item If $d\geq 3$ then for \emph{every} parity-check matrix $\bH$ we
  have $\w^{\min}_{\AWGNC}(\bH)\geq 3$.  This is not true for the BSC
  and the max-fractional weight.
\end{itemize}

These observations show that there is some significant difference
between the various types of pseudocodeword redundancies.



\section{Cyclic Codes Meeting the Eigenvalue
  Bound}\label{sec:kv_bound}

In this section we apply the following eigenvalue-based lower bound on
the minimum AWGNC pseudoweight, proved in \cite{KV-lower-bounds}.

\begin{proposition}\label{prop:KV_bound}
  The minimum AWGNC pseudoweight for a $(w_c,w_r)$-regular
  parity-check matrix $\cH$ whose corresponding Tanner graph is
  connected is bounded below by
  \begin{equation}\label{eq:KV_bound}
    \weight_{\AWGNC}^{\min} \ge n \cdot \frac{ 2w_c - \mu_2 } 
    {\mu_1 - \mu_2 } \; ,
  \end{equation}
  where $\mu_1$ and $\mu_2$ denote the largest and second largest
  eigenvalue (respectively) of the matrix $\cL \define \cH^T \cH$,
  considered as a matrix over the real numbers.
\end{proposition}

We consider now binary cyclic codes with full circulant parity-check
matrices, defined as follows: Let $\cC$ be a binary cyclic code of
length $n$ with check polynomial $h(x)=\sum_{i\in\cI}h_ix^i$ (cf.\
\cite{MacWilliams_Sloane}, p.~194). Then the \emph{full circulant
  parity-check matrix} for~$\cC$ is the $n\times n$ matrix
$\cH=(H_{j,i})_{i,j\in \cI}$ with entries $H_{j,i}=h_{j-i}$.  Here,
all the indices are modulo $n$, so that $\cI=\{0,1,\dots,n-1\}$.

Since such a matrix is $w$-regular, where $w=\sum_{i\in\cI}h_i$, we
may use the eigenvalue-based lower bound of
Proposition~\ref{prop:KV_bound} to examine the AWGNC pseudocodeword
redundancy: If the right hand side equals the minimum distance $d$ of
the code~$\cC$, then $\rho_{\AWGNC}(\cC)\leq n$.

Note that the largest eigenvalue of the matrix $\cL = \cH^T \cH$ is
$\mu_1=w^2$, since every row weight of $\cL$ equals
$\sum_{i,j\in\cI}h_ih_j=w^2$.  Consequently, the eigenvalue
bound is
\[ \weight_{\AWGNC}^{\min} \geq n \cdot \frac{2w-\mu_2}{w^2-\mu_2}
\;, \] where $\mu_2$ is the second largest eigenvalue of $\cL$.  We
remark further that $\cL=(L_{j,i})_{i,j\in\cI}$ is a symmetric
circulant matrix, with $L_{j,i}=\ell_{j-i}$ and
$\ell_i=\sum_{k\in\cI}h_kh_{k+i}$.  The eigenvalues of $\cL$ are thus
given by \[ \lambda_j = \sum_i\ell_i\zeta_n^{ij} =
\operatorname{Re}\sum_i\ell_i\zeta_n^{ij} = \sum_i\ell_i\cos(2\pi
ij/n) \] for $j\in\cI$, where $\zeta_n=\exp(2\pi\bldi/n)$,
$\bldi^2=-1$, is the $n$-th root of unity (see e.g.~\cite{Davis-book},
Theorem~3.2.2).

We also consider quasi-cyclic codes of the form given in the following
remark.
\begin{remark}\label{rem:qc}
  Denote by $\ones_m$ the $m\times m$ matrix with all entries equal to
  $1$.  If $\cH$ is a $w$-regular circulant $n\times n$ matrix then
  the Kronecker product $\tilde{\cH} \define \cH\otimes\ones_m$ will
  be a $w$-regular circulant $mn\times mn$-matrix and defines a
  quasi-cyclic code.  We have
  \[ \tilde{\cL} = \tilde{\cH}^T\tilde{\cH} = \cH^T\cH \otimes
  \ones_m^T\ones_m = \cL \otimes (m\ones_m)\:,\] and the eigenvalues
  of $m\ones_m$ are $m^2$ and $0$.  Thus, the largest eigenvalues of
  $\tilde{\cL}$ are $\tilde{\mu}_1 = m^2\mu_1 = m^2w^2$ and
  $\tilde{\mu}_2 = m^2\mu_2$, and the eigenvalue bound of
  Proposition~\ref{prop:KV_bound} becomes
  \[ \weight_{\AWGNC}^{\min} \ge mn \cdot \frac{2mw-m^2\mu_2}{m^2w^2-m^2\mu_2} = 
  n \cdot \frac{2w-m\mu_2}{w^2-\mu_2} \;. \]
\end{remark}\medskip

We carried out an exhaustive search on all cyclic codes $\cC$ up to
length $n\leq 250$ and computed the eigenvalue bound in all
cases where the Tanner graph of the full circulant parity-check matrix
is connected, by using the following algorithm:

\hrulefill\medskip

\textbf{Input:} Parameter $n$ (code length).

\textbf{Output:} For all divisors of $x^n-1$, corresponding to cyclic
codes $\cC$ with full circulant parity-check matrix, such that the
Tanner graph is connected: the value of the eigenvalue bound.\medskip

\begin{enumerate}
\item Factor $x^n-1$ over~$\ff_2$ into irreducibles, using Cantor and
  Zassenhaus' algorithm (cf.~\cite{vzGG-book}, Section~14.3).
\item For each divisor $f(x)$ of $x^n-1$:
\begin{enumerate}
\item Let $f(x)=\sum_ih_ix^i$ and $\cH=(h_{j-i})_{i,j\in\cI}$.
\item Check that the corresponding Tanner graph is connected (that the
  gcd of the indices $i$ with $h_i=1$ together with $n$ is~$1$).
\item Compute the eigenvalues of~$\cL=\cH^T\cH$: Let
  $\ell_i=\sum_{k\in\cI}h_kh_{k+i}$ and for $j\in\cI$ compute
  $\sum_i\ell_i\cos(2\pi ij/n)$.
\item Determine the second largest eigenvalue $\mu_2$ and output
  $n\cdot(2\ell_0-\mu_2)/(\ell_0^2-\mu_2)$.
\end{enumerate}
\end{enumerate}

\hrulefill\medskip

This algorithm was implemented in the C programming language.
Tables~\ref{table:KV_bound-2} and~\ref{table:KV_bound-1} give a
complete list of all cases in which the eigenvalue bound equals the
minimum Hamming distance $d$, for the cases $d=2$ and $d\ge 3$
respectively.  In particular, the AWGNC pseudoweight equals the
minimum Hamming distance in these cases as well and thus we have for
the pseudocodeword redundancy $\rho_{\AWGNC}(\cC)\leq n$.  All
examples of distance $2$ are actually quasi-cyclic codes as in
Remark~\ref{rem:qc} with parity-check matrix $\tilde{\cH} =
\cH\otimes\ones_2$.  We list here the constituent code given by the
parity-check matrix~$\cH$.

\addtolength{\textheight}{-9.8cm}   

\begin{table}
  \caption{Binary Cyclic Codes up to Length 250 
    with $d=2$ Meeting~the~Eigenvalue~Bound}\label{table:KV_bound-2}
  {\centering
    \begin{tabular}{ccl}
      parameters & $w$-regular & constituent code\vspace{.5mm} \\\hline
      & & \vspace{-2mm} \\
      $[2n,2n\!-\!m,2]$ & $2^m$ & Hamming c., $n=2^m\!-\!1$, $m=2\dots 6$ \\
      \!\!\!$[2n,2n\!-\!m\!-\!1,2]$\!\!\! & 
      \!\!\!$2^m\!-\!2$\!\!\! & Hamming c. with overall 
      p.-check \\
      $[42,32,2]$ & $10$ & projective geometry code $PG(2,4)$ \\
      $[146,118,2]$ & $18$ & projective geometry code $PG(2,8)$ \\
      $[170,153,2]$ & $42$ & a certain $[85,68,\geq\!6]$ $21$-regular code \\
      & & (the eigenvalue bound is 5.2) \\
    \end{tabular}\\}%
\end{table}

\begin{table}
  \caption{Binary Cyclic Codes up to Length 250 
    with $d\geq 3$ Meeting~the~Eigenvalue~Bound}\label{table:KV_bound-1}
  {\centering
    \begin{tabular}{ccl}
      parameters & $w$-regular & comments\vspace{.5mm} \\\hline
      & & \vspace{-2mm} \\
      $[n,1,n]$ & $2$ & repetition code, $n=3\dots 250$ \\
      $[n,n\!-\!m,3]$ & \!\!$2^{m-1}$\!\! 
      & Hamming c., $n=2^m\!-\!1$, $m=3\dots 7$ \\
      $[7,3,4]$ & $3$ & dual of the $[7,4,3]$ Hamming code \\
      $[15,7,5]$ & $4$ & Euclidean geometry code EG(2,4) \\
      $[21,11,6]$ & $5$ & projective geometry code PG(2,4) \\
      $[63,37,9]$ & $8$ & Euclidean geometry code EG(2,8) \\
      $[73,45,10]$ & $9$ & projective geometry code PG(2,8) \\
    \end{tabular}\\}%
\end{table}

We conclude this section by proving a result which was observed by the
experiments.

\begin{lemma}
  Let $m\ge 3$ and let $\cC$ be the intersection of a Hamming code of
  length $n=2^m-1$ with a simple parity-check code of length $n$,
  which is a cyclic $[n,n-m-1,4]$ code.  Consider its full circulant
  parity-check matrix $\cH$.  Then \[ \w^{\min}_{\AWGNC}(\cH)\ge
  3+\frac{1}{2^{m-2}-1}>3 \;. \]

  In particular, if $m=3$ then $\cC$ is the $[7,3,4]$ code and the
  result implies $\w^{\min}_{\AWGNC}(\cH)=4$ and $\rho_{\AWGNC} (\cC) \le n$. 
\end{lemma}

\begin{proof}
  Let $\cH$ be the $w$-regular full circulant parity-check matrix for
  $\cC$.  We claim that $w=2^{m-1}\!-\!1$.  Indeed, each row $\bldh$
  of $\cH$ is a codeword of the dual code $\cC^{\bot}$, and since
  $\cC^{\bot}$ consists of the codewords of the simplex code and their
  complements, the weight of $\bldh$ and thus $w$ must be
  $2^{m-1}\!-\!1$, $2^{m-1}$, or $2^m\!-\!1$.  But $w$ cannot be even,
  for otherwise all codewords of $\cC^{\bot}$ would be of even weight.
  As $w=2^m\!-\!1$ is clearly impossible, it must hold
  $w=2^{m-1}\!-\!1$.

  Next, we show that the second largest eigenvalue of
  $\cL=\cH^T\cH=(L_{j,i})_{i,j\in\cI}$ equals $\mu_2=2^{m-2}$.
  Indeed, let $\bldh_1$ and $\bldh_2$ be different rows of $\cH$,
  representing codewords of $\cC^{\bot}$.  As their weight is equal,
  their Hamming distance is even, and thus it must be $2^{m-1}$.
  Hence, the size of the intersection of the supports of $\bldh_1$ and
  $\bldh_2$ is $2^{m-2}\!-\!1$.  This implies that $L_{i,i}=w$ and
  $L_{j,i}=2^{m-2}\!-\!1$, for $i\ne j$.  Consequently, $\cL$ has an
  eigenvalue of multiplicity $n-1$, namely
  $w-(2^{m-2}\!-\!1)=2^{m-2}$, and thus $\mu_2$ must be $2^{m-2}$.

  Finally, we apply Proposition~\ref{prop:KV_bound} to get
  \[ \w^{\min}_{\AWGNC} 
  \ge (2^m\!-\!1)\,
  \frac{2\,(2^{m-1}\!-\!1)-2^{m-2}}{(2^{m-1}\!-\!1)^2-2^{m-2}} 
  = 3 + \frac{1}{2^{m-2}\!-\!1} \;. \]
\end{proof}

\section*{Acknowledgments}

The authors would like to thank Nigel Boston, Christine Kelley and
Pascal Vontobel for helpful discussions.



\begin{thebibliography}{99}
  
\bibitem{KV-characterization}
  {R. Koetter, W.-C. W. Li, P. O. Vontobel, and J. L. Walker,}
  {``Characterizations of pseudo-codewords of (low-density)
    parity-check codes,''}
  {\em Advances in Mathematics}, vol. 213, pp.~205--229, Aug. 2007. 

\bibitem{KV-long-paper}
  {P. O. Vontobel and R. Koetter},
  {``Graph-cover Decoding and Finite-Length Analysis of Message-Passing 
    Iterative Decoding of LDPC Codes,''}
  accepted for {\it IEEE Trans. Inform. Theory}. 
  Also available as ArXiv report at 
  {\tt\footnotesize arXiv:cs.IT/0512078}.

\bibitem{FKKR}
  {G. D. Forney, R. Koetter, F. R. Kschischang, and A. Reznik},
  {``On the Effective Weights of Pseudocodewords for Codes Defined 
    on Graphs with Cycles,''} 
  in {\it Codes, Systems, and Graphical Models}, IMA workshop,
  Minneapolis, USA, Aug. 1999.  New York, USA: Springer, 2001,
  pp.~101--112.

\bibitem{Schwartz-Vardy}
  {M. Schwartz and A. Vardy},
  {``On the Stopping Distance and the Stopping Redundancy of Codes,''}
  {\it IEEE Trans. Inform. Theory}, vol. 52, no. 3, 2006, pp.~922--932.

\bibitem{Kelley-Sridhara}
  {C. Kelley and D. Sridhara}, 
  {``On the pseudocodeword weight and parity-check matrix redundancy of 
    linear codes''},
  in {\it Proc. IEEE Information Theory Workshop (ITW)}, 
  Lake Tahoe, USA, Sep. 2007.

\bibitem{ISIT-Paper}
  {J. Zumbr\"agel, M. F. Flanagan, and V. Skachek},
  {``On the Pseudocodeword Redundancy''},
  to appear in {\it Proc. IEEE International Symposium on Information
    Theory (ISIT)},
  Austin, USA, June 2010. 
  Also available as ArXiv report at 
  {\tt\footnotesize arXiv:1001.1705}.

\bibitem{Feldman-et-al}
  {J. Feldman, M. J. Mainwright, and D. R. Karger},
  {``Using linear programming to decode binary linear codes,''}
  {\it IEEE Trans. Inform. Theory}, vol. 51, no. 3, 2005, pp.~954--972.

\bibitem{maple-convex}
  {M. Franz},
  {``Convex - a Maple package for convex geometry,''}
  Version 1.1 (2009).
  Available at {\tt\footnotesize 
    http://www-fourier.ujf-grenoble.fr/\~{}franz/convex}.

\bibitem{KV-lower-bounds}
  {P. O. Vontobel and R. Koetter},
  {``Lower Bounds on the Minimum Pseudo-Weight of Linear Codes,''}
  in {\it Proc. IEEE International Symposium on Information Theory (ISIT)}, 
  Chicago, USA, June/July 2004, p.~67. 

\bibitem{MacWilliams_Sloane}
  {F. J. MacWilliams and N. J. A. Sloane},
  {\it The Theory of Error-Correcting Codes},
  Amsterdam, The Netherlands: North-Holland, 1978.

\bibitem{Davis-book}
  {P. J. Davis},
  {\it Circulant Matrices},
  New York, USA: John Wiley \& Sons, 1979.

\bibitem{vzGG-book}
  {J. von zur Gathen and J. Gerhard},
  {\it Modern Computer Algebra} (2nd ed.),
  Cambridge, United Kingdom: Cambridge University Press, 2003.

\end{thebibliography}
\end{document}